\documentclass {journal}

\usepackage {a4wide}
\usepackage {amssymb}
\usepackage [usenames] {color}

\usepackage {enumerate}
\usepackage {plaatjes}
\usepackage {citesort}

\usepackage [usenames] {color}
\usepackage [left,pagewise] {lineno}

\definecolor {infocolor} {rgb} {0.6,0.6,0.6}

\definecolor {sepia} {rgb} {0.75,0.30,0.15}
\everymath{\color{sepia}}

\usepackage {amsmath}
\usepackage {xspace}

\usepackage[noend,linesnumbered]{algorithm2e}

\usepackage{xspace}

\newcommand {\mathset} [1] {\ensuremath {\mathbb {#1}}}
\newcommand {\R} {\mathset {R}}

\newcommand {\etal} {\textit {et al.}}

\newcommand {\locmin} {\ensuremath {\mathrm {lmin}}}
\newcommand {\locmax} {\ensuremath {\mathrm {lmax}}}

\newcommand {\pdismaxconsub} {\textsc {Planar 2-Disjoint Maximally Connected Subgraphs}\xspace}

\newcommand {\pdmcs} {\textsc {P2-MaxCon}\xspace}
\newcommand {\pdisconsub} {\textsc {Planar 2-Disjoint Connected Subgraphs}\xspace}
\newcommand {\pdcs} {\textsc {P2-Con}\xspace}
\newcommand {\disconsub} {\textsc {2-Disjoint Connected Subgraphs}\xspace}
\newcommand {\dcs} {\textsc {2-Con}\xspace}
\newcommand {\setc} {\textsc {Set Cover}\xspace}

\newcommand {\gp} {\pdmcs}

\newcommand{\marrow}{\marginpar[\hfill$\longrightarrow$]{$\longleftarrow$}}
\newcommand{\remark}[3]{\textcolor{blue}{\textsc{#1 #2:}} \textcolor{red}{\marrow\textsf{#3}}}

\newcommand{\maarten}[2][says]{\remark{Maarten}{#1}{#2}}
\newcommand{\rodrigo}[2][says]{\remark{Rodrigo}{#1}{#2}}
\newcommand{\frank}[2][says]{\remark{Frank}{#1}{#2}}

\def\fk#1{%
 {\color{green}{#1}}%
}

\renewcommand{\remark}[3]{}

\newtheorem {theorem} {Theorem}

\newtheorem {lemma} {Lemma}
\newtheorem {observation} {Observation}
\newtheorem {corollary} {Corollary}

\newenvironment {proof}{\textbf {Proof:}}{\hfill \ensuremath {\boxtimes}}

\title{Removing Local Extrema from Imprecise Terrains}

\author
{
  Chris Gray\thanks
  {
Department of Computer Science, TU Braunschweig, Germany, {\tt
gray@ibr.cs.tu-bs.de}
  }
  \and Frank Kammer\thanks
  {
Institut f\"ur Informatik, Universit\"at Augsburg, Germany, {\tt kammer@informatik.uni-augsburg.de}
  }
  \and Maarten L\"offler\thanks
  {
Computer Science Department, University of California, Irvine,
USA, {\tt mloffler@uci.edu}
  }
  \and Rodrigo I. Silveira\thanks
  {
Dept. de Matem\`{a}tica Aplicada II, Universitat
Polit\`{e}cnica de Catalunya, Spain, {\tt
rodrigo.silveira@upc.edu}
  }
}

\begin{document}
\maketitle

\begin{abstract}
  In this paper we consider imprecise terrains, that is, triangulated terrains with a vertical error interval in the vertices. In particular, we study the problem of removing as many local extrema (minima and maxima) as possible from the terrain. We show that removing only minima or only maxima can be done optimally in $O (n \log n)$ time,
  for a terrain with $n$ vertices. Interestingly, however, removing both the minima and maxima simultaneously is NP-hard, and is even hard to approximate within a factor of $O (\log \log n)$ unless $\mathrm{P}= \mathrm{NP}$.
  Moreover, we show that even a simplified version of the problem where vertices can have only two different heights is already NP-hard, a result we obtain by proving hardness of a special case of \disconsub, a problem that has lately received considerable attention from the graph-algorithms community.
\end{abstract}

\tableofcontents

\section{Introduction}

Digital terrain analysis is an important part of geographical information science, with applications in hydrology, geomorphology, visualization, and many other fields~\cite {wg-tapa-98}.
A popular structure for representing terrains is the \emph {triangulated irregular network (TIN)}, also known as \emph {polyhedral terrain}.
In this model, a terrain is represented by a planar triangulation with an additional height associated with each vertex.
If we linearly interpolate the heights of the vertices, we also obtain a height at every other point in the plane, resulting
in a bivariate, piecewise linear and continuous function, defining the
surface of the terrain.
A terrain in this model is also often called a 2.5-dimensional (or 2.5D)
terrain.

\subsection {Imprecision in Terrains}

In computational geometry it is usually assumed that the input data
for any problem is correct and known exactly.
In practice, this is unfortunately not the case.
There are many sources of imprecision, the most prominent of which
is the data acquisition itself.
In terrain modeling, this is particularly relevant, because elevation data is collected by measuring devices that are ultimately error-prone.
Often such devices produce heights with a known error bound or
return a height interval rather than a fixed height value.

In order to handle the imprecision in terrains, we adopt the
model used in~\cite{ge-osp-04,gls-smit-10,ke-opsp-07}, where the
height of each terrain vertex is not precisely known, but only
an interval of possible heights is available.
This results in considerable freedom in the terrain, since the ``real'' terrain is
unknown and any choice of a height for each vertex---as long as
it is within its height interval---leads to a valid \emph{realization}
of the imprecise terrain.
The large number of different realizations of an imprecise
terrain leads naturally to the problem of finding one that is
`best' according to some criterion, or that removes most or many instances of a certain type of unwanted feature (or artifact) from the terrain.

We note that, even though terrain data may contain error also in the $x,y$-coordinates, under this model we  consider imprecision only in the $z$-coordinate.
This simplifying assumption is justified by the fact that error in the $x,y$-coordinates will most likely produce elevation error.
Moreover, often the data provided by commercial terrain data suppliers only reports the elevation error~\cite{ft-ccedem-06}.

In the remainder of this paper, an \emph {imprecise terrain} is a set of $n$ vertical intervals in $\R^3$, together with a triangulation of the vertical projection. Figure~\ref {fig:imprecise-terrain} shows an example.
A \emph {realization} of an imprecise terrain is a triangulated terrain that has the same triangulation in the projection, and exactly one vertex on each interval.
An alternative way to view an imprecise terrain is by connecting the tops of all intervals into a terrain, which we call the \emph {ceiling}, and the bottoms into a second terrain, which we call the \emph {floor}. Then, a realization is a terrain that lives in the space left open between the floor and the ceiling. Figure~\ref {fig:floor-and-ceiling} shows this in the example.

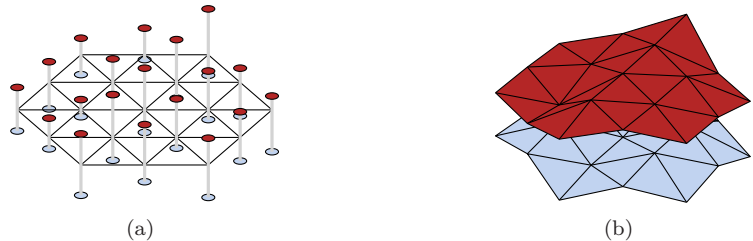
\tweeplaatjes {imprecise-terrain} {floor-and-ceiling} {(a) An example of an imprecise terrain. (b) The same terrain, shown by drawing the floor and the ceiling.}

\subsection {Removing local extrema}

A local minimum (or pit) is a location on a terrain that is surrounded by higher points, or that has no lower neighboring point.
Similarly, a local maximum (or peak) is as a point surrounded by lower points or without higher neighbors. The term \emph{local extrema} will be used to refer to both local minima and local maxima.

When terrains are used for land erosion, landscape evolution, or hydrological studies, it is generally accepted that the
majority of local extrema in the terrain model are spurious, caused by errors in the data or model production.
A terrain model with many pits or peaks does not represent the terrain faithfully, and moreover, in
the case of pits, it can create problems because water accumulates at them, affecting water flow
routing simulations.
For this reason the removal of local minima from terrain models is a standard preprocessing
requirement for many uses of terrain models~\cite{ztz-edpah-06,tsv-adddl-06}.
However, existing preprocessing routines make no attempt to relate the removed
minima to knowledge about the imprecision in the terrain model, possibly
causing major alterations to the data under study.

In this paper we attempt to solve the problems of removing as many local minima, maxima, or extrema as possible by moving the vertices of an
\emph {imprecise} terrain within their allowed height intervals.
The rationale behind this is that if a pit (or peak) can be removed in this
way, it is likely to be an artifact of the data, whereas if it cannot, it
is more certain to be a `real' pit (or peak).
We define the \emph{minimizing-minima}, the
\emph{minimizing-maxima}, and the \emph{mini\-mizing-extrema} problems on imprecise terrains, where we attempt
to find a realization of an imprecise terrain
(by placing the imprecise points within their intervals) that minimizes
the number of \emph{local minima}, \emph{local maxima}, and
\emph{local extrema}, respectively.

It is important to note that a group of $k$ connected vertices at the same height without any lower neighbor is considered to be only \emph{one} local minimum.
This is reasonable from the point of view of the application, and follows the definitions used in previous work~\cite{so-fcdd-07}.
In Section~\ref {sec:degeneracy} we discuss what the implications of this modeling choice are for our results.

Regarding previous work, a lot of research has been devoted to the problem of removing local minima from (precise) terrains, especially in the geographic information science community, but also from more algorithmic points of view;
we only provide a few relevant references here.
Most of the literature assumes a raster (grid) terrain (e.g.~\cite{m-nmg-88,mg-aobat-99,ztz-edpah-06}), and employs methods that are some type of ``pit filling'' technique, which consist in filling in depressions until they disappear (e.g.~\cite{cdp-qafdn-00,m-nmg-88,ztz-edpah-06}).
Some of the few exceptions are the methods in~\cite{cd-bdgt-10,mg-aobat-99,r-pbauca-98}.
A few algorithms have been proposed for triangulated terrains, such as~\cite{ls-ftt-04,aay-ioebu-06}.
The removal of local extrema has also been studied in the context of optimal higher order Delaunay triangulations~\cite{ghk-hodt-02,kkl-grtho-07}.
In particular, Gudmundsson~\etal~\cite{ghk-hodt-02} show that the optimal number of both local minima and local maxima can be removed from first-order Delaunay triangulations in $O(n \log n)$ time.
More related to this paper, Silveira and Van Oostrum~\cite{so-fcdd-07} study moving vertices vertically in order to remove all local minima with a minimum cost, but do not assume bounded intervals.


\subsection {Results}

In Section~\ref {sec:locmin}, we first study the problem of finding a realization of an imprecise terrain that minimizes the number of local minima (or local maxima).
We show that all potential local minima (resp. maxima) of the terrain are independent, that is, whether we remove one does not influence whether or not we can remove another.
Using this property, we then present a relatively simple algorithm that removes local minima (maxima) optimally in $O(n \log n)$ time.

In Section~\ref {sec:extrema}, we turn our attention to removing both minima and maxima simultaneously. In this case we no longer have the independence property, and as a consequence the problem becomes much harder.
In fact, we show that the problem of minimizing the total number of local extrema is NP-hard, even
hard to approximate within a factor $O (\log \log n)$ unless $\mathrm{P}= \mathrm{NP}$.

All results mentioned above assume general position of the input (that is, all the top and bottom ends of the intervals have different heights), and
considers points chosen to be at the same height as a group to be at most
one single local extremum.
In Section~\ref {sec:degeneracy}, we discuss how these assumptions influence the results presented in 
Section~\ref {sec:extrema}.

Finally, in Section~\ref {sec:intermezzo}, we consider a simplified version
of the problem of removing local extrema, where grobally all vertices of a
terrain have only two possible heights. We show that this problem is
already NP-hard, and cannot be approximated within a factor $3/2$.
For this, we prove that the planar version of \disconsub is NP-hard.
The latter problem has received quite some attention recently, and we consider the connection this hardness proof to be of independent interest.

\section {Removing local minima} \label {sec:locmin}



We begin with the problem of finding a realization that has the smallest number of local minima, that is, the minimizing-minima problem.
We propose an efficient algorithm based on the idea of selectively \emph{flooding} parts of the terrain.
The algorithm begins with all
vertices as low as possible, and simulates flooding parts of the terrain.

\paragraph{Algorithm}
Conceptually, we raise all local minima as much as possible, that is, we raise each minimum and its neighbors as we meet them, merging minima as we sweep the terrain bottom-up. The process stops when one of the vertices in a local minimum cannot be raised any further. Also, when there is only one local minimum left and no more higher terrain it could merge into, the process stops.

We sweep a horizontal plane vertically, starting at the lowest interval end and moving upwards in the $z$ direction.
As the plane moves up, it \emph{pulls} some of the vertices with it, whose height is changing together with the plane.
At any moment during the sweep, each vertex is in one of three states:
\begin {enumerate} [(a)]
  \item Moving, if it is currently part of a local minimum, and is moving up together with the sweep plane.
  \item Fixed, at a height lower than the current one.
  \item Unprocessed, if it has not been reached by the sweep plane yet.
\end {enumerate}

As the sweep plane moves vertically up, we distinguish two types of events:
\begin {enumerate} [(i)]
  \item The plane reaches the beginning (lowest end) of the interval of a vertex,
  \item The plane reaches the end (highest end) of the interval of a vertex.
\end {enumerate}

Let $v$ denote the vertex whose interval just began or ended, and let $h$ be the current height of the plane. Note that all fixed vertices are fixed at a height lower than $h$.\footnote{For simplicity we are assuming in this description that all interval heights are different. The removal of this assumption does not pose any problem for the algorithm.}

An event of type (i) can create a number of situations.

If $v$ has a neighbor that is already fixed, then $v$ will never be a local minimum, thus $v$ is fixed at its lowest possible height.
Moreover, if some other neighbor of $v$ is currently part of a local minimum (i.e. is moving), then all the vertices part of that local minimum become fixed at $h$, and automatically stop being a minimum.
This occurs for each neighbor of $v$ that is currently part of a local minimum.

If all neighbors of $v$ are currently unprocessed, then $v$ becomes a new local minimum, and starts to move up together with the plane.

Finally, if no neighbor is fixed but some neighbor is moving, thus is part of a local minimum, then $v$ will join that existing local minimum and also start to move up together with the plane (note that if there is more than one local minimum that is connected to $v$, at this step they all merge into one).

Events of type (ii), when an interval ends, are easier to handle. If $v$ is fixed, nothing occurs.
If $v$ was moving, then it becomes fixed at $h$, and the same occurs to all the vertices of the local minimum that contains $v$.
Thus the whole local minimum becomes fixed, and will be present in the final solution.

\paragraph{Correctness} The correctness of the algorithm can be proved by induction on the steps (i.e. events) of the sweep (associated with exactly $2n$ height values).
Let $h_i$ denote the height of the plane at the $i$th event.
Clearly, for $h=h_1$ the terrain processed has only one local minimum, comprised of the lowest vertex, which is optimal.
Now assume that for $h=h_i$ the solution is optimal.
That is, the number of local minima in the imprecise terrain resulting from cropping the original terrain $T$ by the plane at height $h_i$ (that is, $T \cap (z \leq h_i)$) is minimum.

We analyze the type of event that can take place for $h=h_{i+1}$. Let $v$ be the vertex whose interval is ending or beginning at $h=h_{i+1}$.

If the interval of $v$ is ending, then $v$ is part of a local minimum that will be fixed.
Since this local minimum already existed in the previous step, and that solution was optimal by the inductive hypothesis, the current solution is also optimal.

If the interval of $v$ is just starting, it is only necessary to argue about the optimality of the connected component (induced by the vertices in the cropped terrain) that contains $v$. 
The other connected components are optimal due to the inductive hypothesis, because they have not changed by this event.

Consider first the case in which $v$ is connected to at most one fixed (lower) local minimum
in the current cropped terrain.
Then the connected component that contains $v$ 
will end up consisting of a single local minimum.
Since every connected component has at least one local minimum, this is optimal for the component 
that contains
$v$.

In case that $v$ is connected to more than one lower local minimum, we note that none of them can be removed by connecting them to $v$, because the lowest possible position for $v$ is at $h_{i+1}$, which is higher than all its neighbors (recall that, by construction, fixed local minima are at their highest possible height).
Therefore in this case the number of local minima for the connected component that contains $v$ stays the same, leading again to an optimal solution for that component.
Therefore the current cropped terrain has the minimum possible number of local minima, and the correctness of the algorithm follows.

Finally,  $v$ can become connected to one or more vertices that were moving (thus were part of one or more local minima).
In this case, they all become one single (moving) connected component, connected through $v$, and hence one single local minima.
Again, this is optimal for that component.

\paragraph{Running time}
Sorting the interval ends for the sweep requires $O(n \log n)$ time.
We put the ends in an event queue and mark each event as either of type (i) or type (ii). We remove all events of type (ii) that come after the last event of type (i).

The rest of the steps can be implemented in linear time as follows.
For every vertex, we simply maintain a label that has a value of either \emph {moving}, \emph {fixed} or \emph {unprocessed}.

At an event of type (i) where the sweep plane reaches the bottom of an unprocessed vertex $v$, we first inspect all neighbours of $v$ to determine which subcase we are in. This takes time proportional to the number of neighbours,
and we charge this cost to the edges connecting $v$ to its neighbours. Since
we charge each edge at most twice over the whole algorithm, this takes linear time in total. Now, if no neighbour of $v$ is fixed, we simply set the label of $v$ to \emph {moving} in constant time. If some neighbour of $v$ is fixed, 
we set the label of $v$ to \emph {fixed}, and we start a floodfill (for example using a depth first search) in the graph 
induced by the moving vertices
to find all vertices connected to $v$ that are currently set to \emph {moving}; we set them to \emph {fixed} as well, and we set their height to the current height of the sweep plane. This takes time proportional to the number of vertices that are being fixed plus the number of edges connecting these vertices to other vertices. Since each vertex gets fixed only once, this also amounts to linear work in total.

Events of type (ii) are handled similarly. If $v$ was moving, we set its label to \emph {fixed} and also start the floodfill in the same way.

When all events in the queue have been processed, we finally set the remaining moving vertices to fixed as well.

Note that by multiplying all interval ends with $-1$, we can solve also the
minimizing-maxima problem.

\begin {theorem}
The minimizing-minima (or minimizing-maxima) problem in an imprecise terrain with $n$ vertices can be solved in $O(n \log n)$ time.
\end{theorem}

It is interesting to note that when a group
of $k$ connected vertices at the same height without any lower
neighbors is regarded as $k$ different local minima, the problem can be proved NP-hard.
More details on this are given in Section~\ref {sec:degeneracy}.

\section {Removing all local extrema}
\label{sec:extrema}
We now move on to the problem of removing 
all local extrema at the same time.
Although the algorithm in the previous section works for both removing minima and removing maxima, it is not possible to use both height assignments simultaneously.
We will show in the next section that we can still use the algorithm twice to narrow down the problem, without changing the value of the solution.
Unfortunately, such an approach does not help much 
to find an optimal realization minimizing the number of maxima. In Section~\ref{sec:ReductionSetCover}  we give a proof that shows that minimizing-extrema is NP-hard to approximate within a factor of $O (\log \log n)$, using a reduction from \setc.

\subsection {Canonical form of an imprecise terrain}

Recall that the floor $F$ is the realization formed by all lower endpoints of the imprecise vertices, and the ceiling $C$ is the realization formed by all upper endpoints, as shown in Figure~\ref {fig:floor-and-ceiling}.

Given two realizations $X$ and $Y$ of the same imprecise terrain, we use the notation $(X,Y)$ to refer to the imprecise terrain truncated by $X$ and $Y$: the bottom interval ends are taken from the heights in $X$, and the top interval ends  from the heights in $Y$ (we assume here that $X$ is never above $Y$).

 We are searching for a surface between the
floor and the ceiling that optimizes the number of local extrema.
We will run the algorithm in Section \ref{sec:locmin} on $(F,C)$ to remove the local minima, 
and call the result $F'$, and run it again on $(F,C)$ to remove local maxima and call the result $C'$.
We call the imprecise terrain $(F', C')$ the \emph {canonical form} of $(F, C)$.

\begin {lemma} \label {lem:proek}
The imprecise terrain induced by $(F', C')$ is still a valid imprecise terrain which has the same optimal solution for removing local extrema as the original terrain $(F, C)$.
\end {lemma}

\begin {proof}
  We need to show two things. To show that $(F', C')$ is still a valid imprecise terrain, we need that the height of any vertex in $F'$ is at least the height of that vertex in $C'$.
  If there exists a height $h$ such that the entire floor lies below $h$ and the entire ceiling lies above $h$, then neither of them will ever raise/lower beyond $h$, because of the stop condition when there are no new interval events anymore.
  Otherwise, if a vertex $v$ does not rise higher this is because its plateau hits a point of the ceiling; clearly a plateau lowering $v$ will never move past this point.
  In both cases, a plateau rising the floor and one lowering the ceiling of the same vertex never cross each other.

  To show that $(F', C')$ has the same optimal solution as $(F, C)$, we need to show that there exists an optimal terrain $T^*$ between $F$ and $C$ that in fact also lies between $F'$ and $C'$.
  This is true because if a terrain would have a local minimum below $F'$, we could freely lift it together with its neighbors until it coincides with $F'$. By the construction of $F'$, we never hit the ceiling during this process, so we never increase the number of minima or maxima. The converse is true for 
   local maxima above $C'$.
\end {proof}

Lemma~\ref {lem:proek} implies that we can run the algorithm of Section~\ref {sec:locmin} as a preprocessing step, while still allowing a solution as good as in the original problem.
Furthermore, the canonical terrain has more structure than the original one. Every remaining local minimum of the floor touches the ceiling, and every remaining local maximum of the ceiling touches the floor. We can show the following:

\begin {lemma}
  The total number of extrema in the optimal solution $T^*$ is never greater
than the number of local maxima of the floor $F'$ + the number of local
minima of the ceiling $C'$.
\end {lemma}
\begin {proof}
  Consider the solution $T = C'$.
  Any local maximum of $T$ must touch a unique local maximum of $F'$ because when lowering the maxima of $C$ this was exactly the condition on which we stopped. Therefore, the number of local maxima of $T$ is smaller than the number of local maxima of $F'$.
  Thus, the number of local extrema in $T$ is at most the number of local maxima of $F'$ + the number of local minima of $C'$ (which is $T$ itself).
  Clearly, the number of local extrema in the optimal solution $T^*$ can only be even smaller.
\end {proof}

If we denote by $\locmin(T)$ the number of local minima in a terrain $T$ and by $\locmax(T)$ the number of local maxima in $T$, and we denote by $T^*$ the optimal solution of our problem $(F, C)$, then we can summarize these observations as follows:

\begin {equation} \label {eq:gap}
  \locmin(F') + \locmax(C') \leq \locmin(T^*) + \locmax (T^*) \leq \locmin (C') + \locmax (F')
\end {equation}

This formula gives a bound on the values of a particular instance.
In theory, the gap may still be arbitrarily large.
For example, consider an instance where the floor is more or less flat except for a number of ``stalagmites'' that reach all the way to the ceiling, and the ceiling is more or less flat except for a number of ``stalactites'' that reach all the way to the floor. Figure~\ref {fig:spikes2} show such a situation.
\maarten {I liked the old spikes figure a little better than this one, but I don't mind too much either way. (Which one is "correct" depends on whether we look at the surface as being smooth or having a sharp angle at the base of the cones. I was originally going for the smooth look.)}
In this case, the number of maxima of the floor and minima of the ceiling is large, while the floor has only a single minimum and the ceiling has only a single maximum, and the preprocessing step will not make a difference.
We see in the next section that this makes the problem very hard to solve.

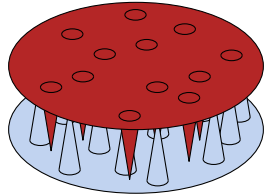
\eenplaatje {spikes2} {An imprecise terrain that has many maxima on the floor and many minima on the ceiling.}

On the other hand, such terrains seem unlikely to appear in real applications. It may be likely that in practice, the gap in Equation~\ref {eq:gap} is quite small. Under which properties of terrains this is the case remains an interesting open question.

  \subsection {Hardness of approximation} \label {sec:approx}
\label{sec:ReductionSetCover}

    In this section we show
    by a direct reduction from the \emph{\setc} problem
    that we cannot approximate the number of local extrema on $n$-vertex graphs
    within any factor better than $O (\log \log n)$ unless $\mathrm{P}= \mathrm{NP}$.

    Given a tuple $(U,C)$, where $U$ is a finite set called
    \emph{universe} and $C$
    is a collection of subsets of $U$ with $\bigcup_{S\in C}S=U$,
    a \emph{set cover} for $(U,C)$ is a collection $C' \subseteq C$ such that
    the union of all sets in $C'$ is equal to $U$. The \emph{size} of $C'$
    is its cardinality.
    The \setc problem is to find a set cover of minimal size.

%

Let $(U,C)$ be an instance of the \setc problem.
We start by defining a graph $G$ with colored vertices.
We then construct a terrain by embedding $G$ in the plane, and triangulating its faces with more than 3 incident vertices.
Finally, we assign heights to the terrain vertices, where the height
of every vertex depends on its color.

    In the remainder, we will use the terms \emph {west} and \emph {east} to refer to the negative and positive $x$-direction, \emph {south} and \emph {north} to refer to the negative and positive $y$-direction, and \emph {down} and \emph {up} to refer to the negative and positive $z$-direction.
    \maarten {Added this sentence here, just before the first usage, because we don't really have a "notation" or "termilonogy" section.}

We begin by describing the northern edge of the constructed graph, which consists in a \emph{north gadget} depicted in Figure~\ref{fig:top}. For each item
$x$ of the universe $U$, we introduce $|U|+3$ red vertices
with a blue vertex between each pair of red vertices. All blue vertices are connected to another
blue vertex $v^{\mathrm{min}}$ at the north of the construction.
Each vertex $v\neq v^{\mathrm{min}}$ in the north gadget
is the beginning of a path---that we call {\em southward path} (indicated in the figures by dashed arrows).
\frank{Here we define southward path?}
\rodrigo{Only partially, right? Because new paths start at gadgets too.}
\frank{Ok, you are right. But lets say here that this is a southward path and
let us say later that this is also a southward path. Of course, more
precisely is to say that, if there is a $v$ without a neighbor that is north to $v$ and
higher than $v$, then $v$ is the starting point of a southward path.
Moreover, all southward path end in a different vertex in the south gadget.}
\rodrigo{OK, we can leave it as it is now, I think that with the new pictures it is clearer anyway}
Moreover, southward paths are marked as either \emph{covered} or \emph{uncovered}.
At this stage, all southward paths starting with a red vertex are considered uncovered, and the ones starting with a blue vertex are covered. 
\maarten [asks] {Why are the arrows in Figure~\ref {fig:top} that connect to red vertices not dashed? Do we consider them to be southward paths as well, or not? (The text suggests yes.) Also, it looks like by "down" do we mean down in the $y$-direction, not down in the $z$-direction. If so, maybe we should say that. Or maybe use "south" as in the other paper, just since it can be confusing when we talk about heights and minima and such...}
\frank{South is a good idea; can you change it? I try to make the lines more dashed, the
problem is that the lines are soo short, i.e., I probable make the lines a
bit longer.}
\maarten {Ok, I did a global search for the words "top", "up", "bottom", "down", "left" and "right" and replaced them with north, south, east and west were appropriate. Are there any other words that we should change?}

\eenplaatje {top} {North gadget.}

The construction continues by adding one \emph{row-gadget} for each set $S \in C$ (see Figure~\ref{fig:row}
for a schematic representation).
Each row-gadget extends the southward paths
southwards and consists of a row of vertices that we call the {\em decision row}. The westernmost and
easternmost vertex of each decision row are, in fact, the same---the
edges connected to these vertices meet in the space above the northern edge of
the currently-constructed graph.
Every decision row consists of white vertices that must be assigned
a color.
To achieve the behavior needed for the reduction, we need to ensure that the colors assigned to the white vertices alternate
between blue and non-blue (red or yellow).
In order to do that, we place inverter gadgets
as shown in Figure~\ref{fig:gadget-a}
between
every pair of white vertices in a decision row.
Additionally, each row-gadget contains one yellow vertex $y_S$ and
several subgadgets that we describe next.

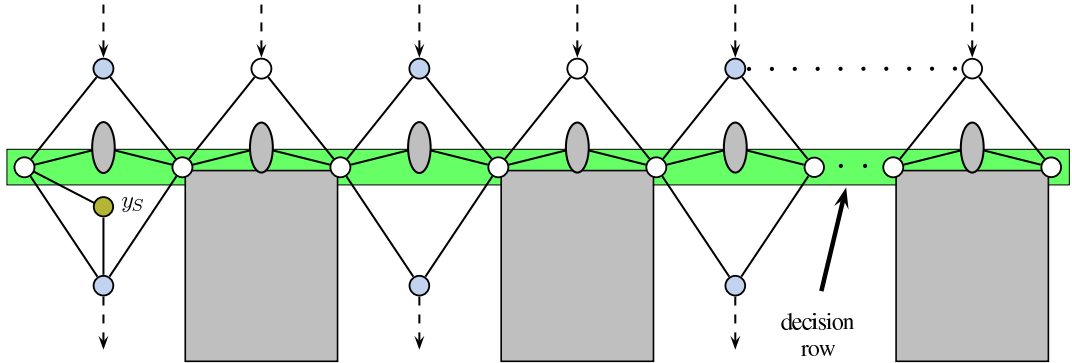
\eenplaatje {row} {Row-gadget for a set $S$ of $C$. Each gray box contains a subgadget.}

As shown in Figure~\ref{fig:row}, the subgadgets always have a white vertex above them.
Depending on the situation of that white vertex, we distinguish two different kinds of subgadgets.

If the white vertex above a subgadget is part of an uncovered southward path
that was introduced for an item $x \in S$, we  use
the subgadget of Figure~\ref{fig:gadget-c}.
The southward path going through the white vertex above such subgadget either
continues its way by the west or east white vertex in the subgadget.
 The
other white vertex is the first vertex of a new southward path for item $x$
 so that we say that both southward paths leaving the subgadget of
 Figure~\ref{fig:gadget-c} are introduced for item $x$.
We mark the west southward path as covered, and the
east southward path as uncovered.
%
As we show later, a
southward path for an item $x$ starting in the north gadget can use the west southward path if and only if
the second west, the forth west, etc.\ vertex of a decision row is not blue. But then
the westernmost, the third west, etc.\ vertex must be blue and the only yellow vertex $y_S$ of the row gadget of $S$
is connected to only two blue vertices, i.e., we have a local maximum at that vertex.
Roughly speaking, this local maximum in the row gadget allows us to mark as covered all southward paths of items $x\in S$.

Otherwise, we use
the subgadget of Figure~\ref{fig:gadget-b}.
A southward path exits this subgadget 
marked as
covered if and only if it
entered this subgadget marked covered.



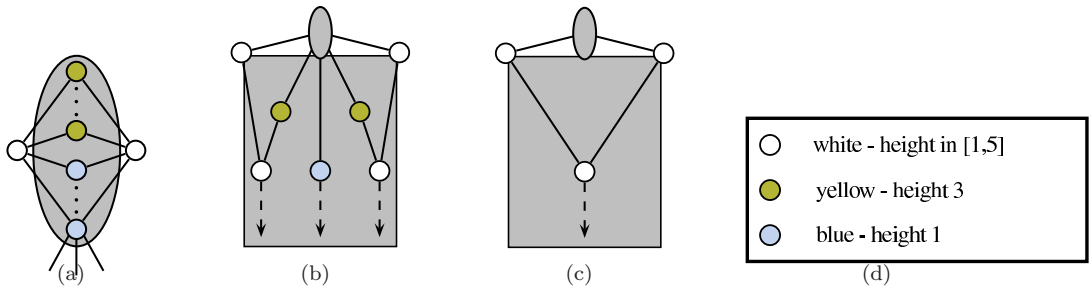
\vierplaatjes {gadget-a} {gadget-c} {gadget-b} {gadget-d} {(a) Inverter gadget consisting of $|U|+3$ blue and yellow vertices.
(b) and (c): Subgadget part of a gadget for a set $S$ in $C$. (d) The heights of
the vertices in Figures (a)-(c).}

The construction ends with a \emph{south gadget}, which is more or less symmetric to the north gadget, see Figure~\ref{fig:bottom}. The lowest
vertex $v^{\mathrm{max}}$ is red.
The color of the vertices with the two colors in Figure~\ref{fig:bottom} is decided with the following rule.
If such a vertex is the end of a path that is marked uncovered,
then it is colored yellow.
Otherwise, it is colored red.

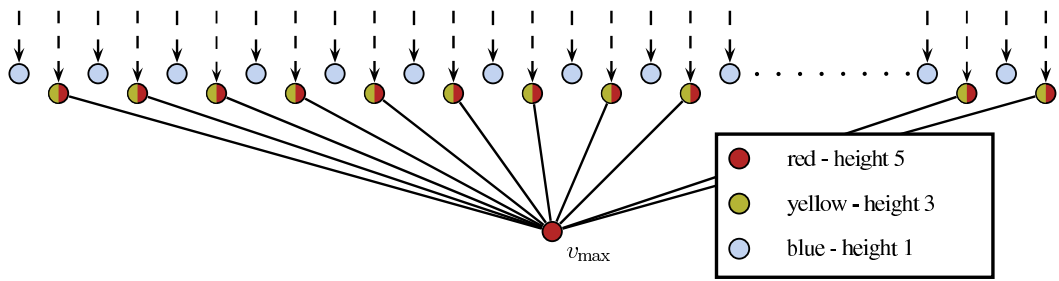
\eenplaatje {bottom} {South gadget.  A two-colored vertex is colored either red or yellow depending on whether it is the endpoint of a path marked covered or uncovered.}

Let $G$ be the graph obtained, with a straight-line embedding $\varphi$.
The heights are assigned to vertices as follows.
For simplicity, we use colors to refer to vertices with the same (imprecise) height.
Vertices colored red
have height
$5$, yellow vertices
have height $3$,
blue vertices have height $1$, and
white vertices have a height in the range $[1, 5]$.
To
triangulate $G$, we add a vertex $v_F$ of height $2$ into each face
$F$ of $\varphi$ and connect $v_F$ to all
vertices adjacent to $F$ in $\varphi$. 
Let $V'$ be the set of vertices added
during the triangulation.
It is easy to verify that all faces of $G$ have at least one blue and one yellow/red vertex, 
thus the vertices in $V'$, at height $2$, cannot be local extrema.

To see that the size of our reduction is polynomial, we must show that the
number $n$ of vertices of $G$ is polynomial in $|C|$ and $|U|$ since
this graph is planar and thus $|V'|$ is linear in $n$.
\begin {lemma}
  The number of vertices in $G$ is $n = O (|C|^2|U|^3)$.
\end {lemma}
\begin {proof}
The northernmost row contains $O(|U|^2)$ vertices, which are the starting points of $O (|U|^2)$ uncovered southward paths.
Then, at every row, each uncovered path may split into a covered path and an uncovered path. Covered paths do not split further. 
This means that in each row we have $O (|U|^2)$ uncovered southward
paths and that the total number of paths increases by at most $O (|U|^2)$ in each row, so in the last row the number of paths is at most $O (|C||U|^2)$.

Now, note that 
an inverter gadget consists of $O (|U|)$ vertices.
Moreover,
a row-gadget $H$ for a set being crossed by $z=O (|C||U|^2)$ southward paths
consists, for each southward path, of an inverter gadget and a constant number of further
vertices. Thus, $H$ has $O(z |U|)=O (|C||U|^3)$ vertices.
Since the north gadget and the south gadget has fewer vertices and since we
have $|C|$ row-gadgets, the total number of vertices is $n = O (|C|^2|U|^3)$, as claimed.
\end {proof}

To minimize the number of local extrema we need to assign a height to each white vertex such that
all blue vertices form one connected component and all
red vertices form one connected component.  Moreover, every connected components of yellow vertices
needs to be connected to both a
blue and a red vertex.

\begin{theorem}
The minimizing-extrema problem in an imprecise terrain with $n$ vertices cannot be approximated
within a factor of $O (\log \log n)$ in polynomial time, unless $\mathrm{P}=\mathrm{NP}$.
\end{theorem}

\begin{proof}
We first show
that each instance $I_1$ for the \setc problem of optimal cost $z-2$ is reduced
to an instance $I_2$ for the minimizing-extrema problem of optimal cost $z$ such
that each a solution for $I_2$ of cost $y$ can be easily
transformed into a solution for $I_1$ of cost $y-2$.

Let $C'\subseteq C$ be a set cover. A coloring of the graph can be found as follows.
In each decision row,
we color the vertices alternating
in blue and non-blue (red or yellow).
We choose for a vertex $v$ between yellow and red depending
on whether $v$ is
incident (from above) to a red vertex.
If so, color $v$ red, and otherwise yellow.
Color the white vertices in the gadget of Figure~\ref{fig:gadget-c}
by the same rule, in red or yellow.
Additionally, we color the westernmost vertex in a decision row blue if and only if
the decision row is part of a gadget for a set $S\in C'$. See
Figures~\ref{fig:inSC} and~\ref{fig:outSC} for an sketch of this coloring.

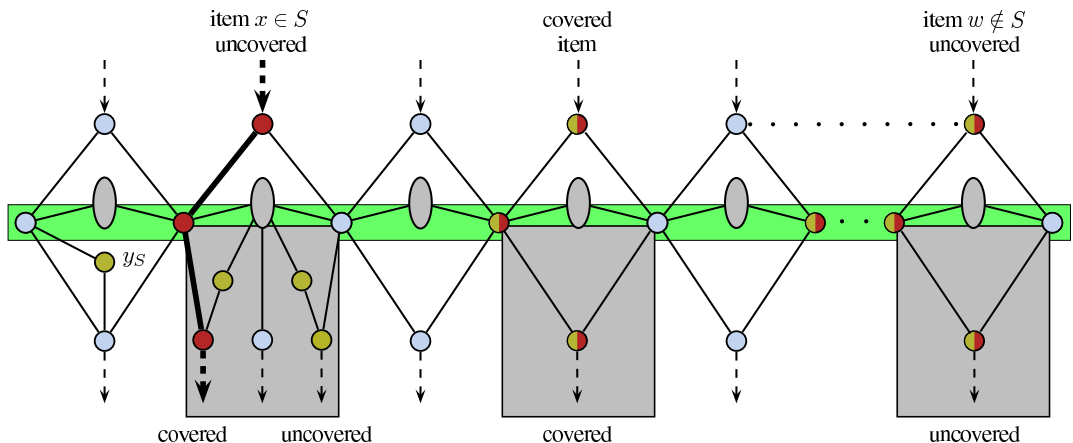
\eenplaatje {inSC} {A set $S\in C'$ with a local maxima at $y_S$ allows to 
switch a red southward path from uncovered to covered.}

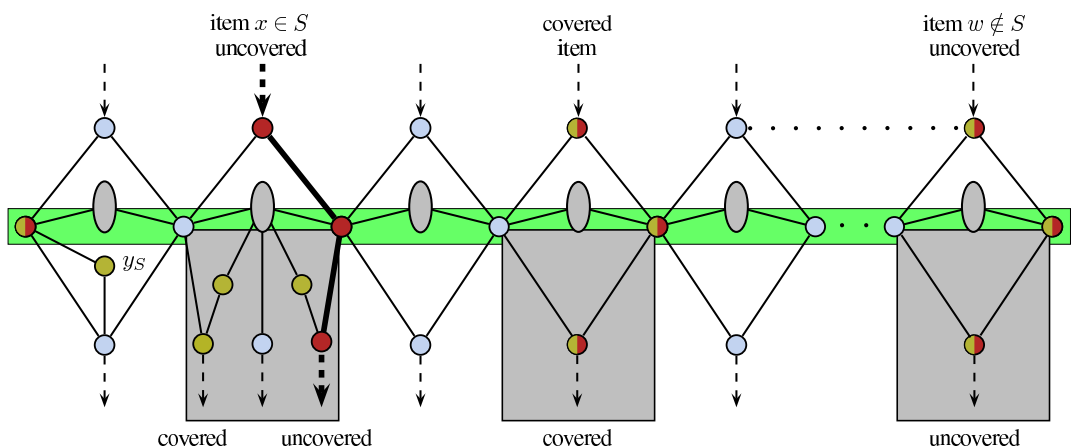
\eenplaatje {outSC} {A set $S\notin C'$ with no local maxima at $y_S$
forbids to 
switch a red southward path from uncovered to covered.}


In this coloring, all vertices in $V'$ are connected to both a blue vertex and either a red or
yellow vertex. Ignoring the initially yellow vertices $y_S$,
all
yellow components are connected to a blue vertex and a
red vertex; this implies that yellow vertices cannot be
minima or maxima.
Moreover, all blue vertices form one connected component.
By our choice of
the height of two-colored vertices in the south gadget,
the red vertices also induce a connected component
since a path starting from a red vertex in the north gadget
contains only red vertices and ends in the southernmost row in a covered
vertex, which is by construction red.
Apart from $v^{\mathrm{min}}$ and $v^{\mathrm{max}}$,
the only local extrema that we have are created by the yellow vertices
in the row-gadgets introduced for each
set in $C'$.
In other words, we have exactly $|C'|+2$ local extrema.

For the converse, let us first consider
the case in which we have a solution with at least $z=|U|+2$ local extrema.
Then, a set cover of size $z-2=|U|$ exists by definition.
Now, let us assume that we have a solution of cost $z<|U|+2$.
Then the vertices in each decision row are alternately colored in
blue and non-blue, and
there is a path for each item
$x$ of
the universe $U$ that is connected to $v^{\mathrm{max}}$.
Our construction then implies that there is a set $S$ in $C$ with $x\in S$ such that the row-gadget
for $S$ has a local maximum
at its yellow vertex $y_S$.
If we choose $C'$ as the collection containing all sets $\fk{S}$ of $C$ whose row-gadget contains a local maximum
at its yellow vertex $y_S$, then $C'$ is a cover of size
at most $z-2$.

Alon, Moshkovitz, and Safra~\cite{AloMS06}
showed that, for an appropriately
chosen constant $c'>0$,
there is no
polynomial-time approximation
algorithm of 
ratio $c' \ln |U|$ for the \setc problem
with universe $U$ unless $\mathrm{P}= \mathrm{NP}$.
Since each set-cover instance with an optimal solution of cost $1$
can be solved to optimality in $|U|^{O(1)}$ time, there can not exist a polynomial-time
approximation of approximation ratio $c' \ln |U|$ for the \setc
problem restricted to instances with
optimal solutions of cost at least $2$ unless $\mathrm{P}= \mathrm{NP}$.

Assume for contradiction that an approximation
algorithm exists for the minimizing-extrema problem in an imprecise terrain with $n$ vertices
with an approximation ratio $(c'/8) \ln \ln n$.
This means, each instance of the minimizing-extrema problem of
optimal cost $x'$ can be solved with cost at most $((c'/8) \ln \ln n) x'$.
By our reduction from above, each set-cover instance $I_1=(U,C)$
with optimal cost $x \ge 2$ can be first transformed into an instance
$I_2$ of the minimizing-extrema problem with $n \le |C|^3|U|^2\le
(2^{|U|})^3 |U|^2 \le 18^{|U|} $
vertices and with optimal cost $x'=x+2\le 2x$.
By our assumption, we can solve $I_2$ such that the obtained solution has
 cost at most $ ((c'/8) \ln \ln n) x'\le
((c'/4) \ln \ln n) x$,
that is, we can solve $I_1$ with cost at most
$((c'/4) \ln \ln n) x -2 \le ((c'/4)  \ln \ln 18^{|U|} ) x \le (c' \ln |U|) x
$---for the latter inequality we use the fact that the cost of an optimal solution is at
least 2, i.e., $|U|\ge 2$.
This is a contradiction to the last paragraph.
Thus,
there cannot exist an approximation
algorithm for the minimizing-extrema problem
in an imprecise terrain with $n$ vertices
with an approximation ratio
$(c'/8) \ln \ln n$.
\end{proof}

Note that Kumar, Arya, and Ramesh~\cite{KumAR00} showed that
the \setc problem with universe $U$ cannot be approximated
within a factor of $o(\log |U|)$
in random polynomial time unless $\mathrm{NP} \subseteq
\mathrm{ZTIME}(n^{O(\log \log n)})$
even if we restrict
the collections $C$ such that $|S_1 \cap S_2|\le 1$ for all $S_1\neq
S_2$ in $C$.
This restriction to $C$ means that
$|C| \le |U|^2$.
To see this, let $U = \{u_1,\ldots, u_{|U|}\}$, and consider
a subcollection $C_1$ of $C$, whose subsets all contain $u_1$.
Then it is easy to see that the number of sets in $C_1$ with
a fixed
$u'\in U\setminus \{u_1\}$ in them is at most $1$. Thus, $|C_1|\le |U|$.
Applying the same argument to all subcollections $C_2,\ldots C_{|U|}$ that have
$u_2,\ldots,u_n$, respectively, in them we conclude that $|C| \le |U|^2$.

Using the reduction from above, the graph obtained for such a restricted set cover instance
$(U,C)$ has $n'=|U|^{O(1)}$ vertices; thus,
we can conclude the following:

\begin{corollary}
The minimizing-extrema problem
in an imprecise terrain with $n$ vertices
cannot be approximated
within a factor of $o (\log n)$ in random polynomial time, unless $\mathrm{NP} \subseteq
\mathrm{ZTIME}(n^{O(\log \log n)})$.
\end{corollary}

\section {Height Degeneracy} \label {sec:degeneracy}

  We have shown that removing local minima is easy and removing local extrema is hard. However, some of our results are dependent on degeneracy issues.
  In this section we will describe how these issues influence the results.

  There are two separate aspects to be considered. One is how we treat vertices of the same height in a realization of an imprecise terrain (that is, in a precise terrain). The other is whether we allow the tops and bottoms of the intervals in the imprecise terrain to have duplicate heights, that is, whether the input is assumed to be in general position.

  \subsection {Minimizing-minima is sometimes hard}

    In Section~\ref {sec:locmin}, we have shown that all local minima can be removed from an imprecise terrain in $O (n \log n)$ time.
    However, this result is based on the viewpoint that when a group of vertices all have the same height, we count them as a single minimum. This is a common viewpoint in the literature, and it is very reasonable from the application point of view. Nonetheless, we show here that if we count them as individual local minima, and the input is not in general position, the problem becomes hard.

    The reduction is from maximum independent set on planar graphs. Given an instance of maximum independent set (a planar graph), we build an imprecise terrain as follows. Each vertex of the graph is replaced by an imprecise vertex with minimum height 1 and maximum height 5. Then, each edge is replaced by $k$ vertices at fixed height 3 in the middle of the edge, which are connected to both neighboring vertices.
    Figure~\ref {fig:flat-in+flat-out} shows an example.
    Finally, we complete the triangulation by adding dummy vertices at height 5 and triangulating the resulting point set.

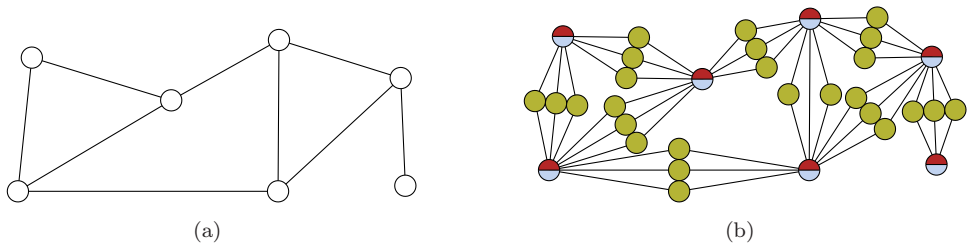
    \tweeplaatjes {flat-in} {flat-out} {(a) A planar graph (b) The white vertices are imprecise vertices at height $[1, 5]$, the edges have been replaced by groups of $k$ yellow vertices, which are fixed at height 3.}

    Now, a vertex becomes a local minimum if it is assigned a height of at most 3. On the other hand, a fixed vertex (an edge of the original graph) becomes a local minimum if both neighbors are assigned a height of at least 3. If $k$ is large enough, this implies that we can only make one of the two vertices incident to an edge higher than 3. The number of local minima in the final terrain is equal to the number of such vertices with a height higher than 3. These vertices form an independent set in the original graph.

    \begin{theorem}
      Minimizing the number of local minima in an imprecise terrain is NP-hard if a local minimum is defined as a vertex without any lower neighbors.
   \end{theorem}

    Note, however, that this proof also relies on degeneracy in the input. If we assume that all input heights  are different, then the edges incident to a given vertex have different heights, meaning we can put the vertex higher than some of them but lower than others. In this situation, the algorithm from Section~\ref {sec:locmin} (with some small adaptations) can still be used to remove all local minima.

  \subsection {Minimizing-extrema is still hard when all heights are different} \label {sec:deghard}


  We can adapt the construction in Section~\ref {sec:approx} to use different heights at all vertices. The reason is that in the final solution of that problem, all paths of red vertices are routed south towards a single high vertex at the southern edge of the construction, while all paths of blue vertices are routed north towards a single low vertex at the northern edge of the construction. This means we can alter the construction, replacing all points $(x, y, z)$ by a point $(x, y, z - \varepsilon y)$. If $\varepsilon$ is small enough, this will make all points with different $y$-coordinates have different $z$-coordinates too, while not changing any property of the construction
(if  before two neighboring vertices had the same height, in the modified construction the southern one will be higher than the northern one).
  Finally, we can make the points with different $x$-coordinates have different heights as well by simply adding some random noise (even smaller than $\varepsilon$).

  \begin{theorem}
  The minimizing-extrema problem in an imprecise terrain with $n$ vertices, for terrains in general position,
  cannot be approximated in polynomial time within a factor of $O (\log \log n)$, unless $\mathrm{P}=\mathrm{NP}$.
  \end{theorem}


\section {Relation to Splitting Graphs}
\label{sec:intermezzo}

  In this section, we explore a relation between the problem of removing local extrema from imprecise terrains and a graph problem which we call \pdismaxconsub (or \pdmcs for short), which we will define shortly. In particular, this shows that removing local extrema is already NP-hard for a very restricted type of terrain: one where vertices can only be one of three types: \emph {low}, \emph {high}, or \emph {unknown}, the last meaning they could be either high or low.

  \subsection {\pdmcs}

  This problem is a special case of \disconsub (or \dcs for short). In that problem, one is given a graph $G = (V, E)$ and two
  disjoint subsets $R \subset V$ and $B \subset V$ of vertices that are colored red and blue. The objective is to find two
  disjoint subsets $R' \supset R$ and $B' \supset B$ such that both $(R', E)$ and $(B', E)$ are connected graphs, that is, to color some of the remaining vertices red or blue to make both the red and the blue subgraph connected.

  The \dcs\ problem has received quite some attention lately. Van 't Hof~\etal~\cite {hpw-pgcp-09} showed that \dcs\ is already NP-hard when there are only two red vertices. Paulusma and Van Rooij~\cite {pr-pgtcs-09} try to tackle the problem by designing more efficient exact algorithms.
  Kammer and Tholey~\cite {kt-cmcc-08} study a related problem called
  the {\em restricted convex coloring problem} where one can color initially uncolored vertices and where one can additionally  uncolor initially colored vertices. This allows, e.g., to handle corrupt data. The results in \cite {kt-cmcc-08} are restricted to  graphs of bounded treewidth.
  To our knowledge there are no results on the \dcs\ problem for planar graphs.

  We define \pdisconsub (or \pdcs for short) as the same problem as \dcs, except that $G$ is known to be planar. We also define \pdismaxconsub  as the optimization variant, where the goal is to optimise the number of connected components (red and blue together) in the output graph, rather than to require that both graphs are completely~connected.

  The idea of the reduction from minimizing local extrema is to take an instance of \pdmcs, and construct from it an imprecise terrain similar to the one shown in Figure~\ref {fig:spikes2}, by replacing the red vertices by stalactites and the blue vertices by stalagmites, and the white vertices by open space.
  We describe the reduction in detail in the next subsection.
  We then proceed to show that \pdmcs is NP-hard, and finally extend the proof to show that \pdcs is also NP-hard.

  \subsection {\pdmcs reduces to minimizing-extrema}
\label{sec:ReductionToOtherProblem}

    We now show that the relation between the problem of removing local extrema from imprecise terrains and the graph problem \pdmcs\ implies that minimizing the number of local extrema in an imprecise terrain is NP-hard. The idea is to take an instance of \pdmcs, and construct from it an imprecise terrain similar to the one shown in Figure~\ref {fig:spikes2}, by replacing the red vertices by stalactites and the blue vertices by stalagmites, and the white vertices by open space.

In our reduction, we take the input to \gp---a planar graph with red, blue and white vertices---and build an imprecise terrain from it.  We will first embed the graph in the plane with straight edges.
We then turn all red vertices into precise vertices at height $5$, and all blue vertices into precise vertices at height $1$.  Finally, we turn the white vertices into imprecise vertices with interval $[1,5]$.

The problem of minimizing extrema on this graph is equivalent to that of minimizing connected components after recoloring.
This is due to the fact that the only way to remove local minima in this terrain is by connecting the minima to each other by putting the white vertices at height 1.  Similarly, the only way to remove maxima is to put the white vertices at height 5.

However, to have a proper imprecise terrain, we must triangulate the graph.  To do this, we add extra vertices and edges as shown in Figure~\ref{fig:reduction-out}.  This adds a component with three extra minima and maxima per inner face of the graph.
In this way, all previous (red, blue, white) vertices are connected to new vertices at height $3$.
These cannot help the red or blue vertices to stop being extrema, and at the same time cannot be extrema because they are connected to a lower and higher vertex inside the component (at heights 2 and 4).
Given a solution that minimizes the number of extrema in the imprecise terrain that we have constructed, we can color the white vertices either red or blue.
For any white vertex whose height is set to $1$, we set the color to blue, otherwise we set the color to red.

      \tweeplaatjes {reduction-in} {reduction-out}
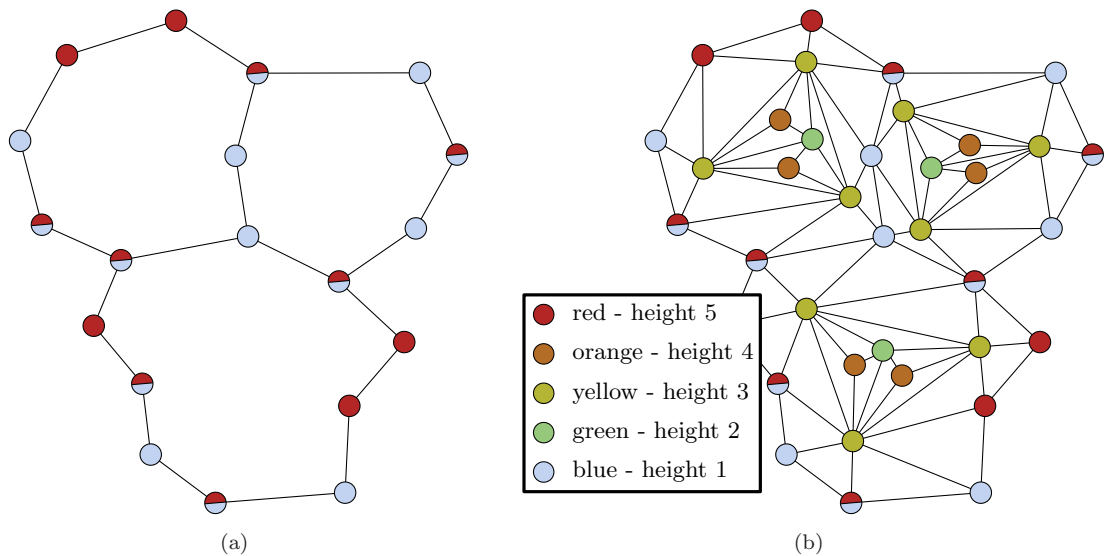
      {(a) An instance of \gp. (b) In the output, we fixed the red vertices at height $5$, and the blue at height $1$. The vertices with two colors are imprecise vertices with interval $[1,5]$. The rest of the vertices are added to make sure that the graph is triangulated, and that the new vertices do not interfere with the number of local extrema. Numbers indicate heights.}

%



\subsection {\gp\ is NP-hard}

We prove that \gp is NP-hard by a reduction from planar
3-SAT~\cite{l-pftu-82}.  In this problem, the normal 3-SAT problem is restricted so that the bipartite graph connecting variables and clauses is planar.  We call this graph \(G_S = ((V \cup C), E)\), where an edge \(e = (v, c) \in E\) if 
and only if variable \(v\) is in clause \(c\).  As usual in such reductions, we first embed \(G_S\) in the plane so that none of the edges in \(E\) cross.  We then replace the vertices and edges in the embedding with ``gadgets''.

The variable gadget is simply a white vertex.  We show below that coloring the vertex red is equivalent to setting the corresponding variable to \texttt{true} and coloring the vertex blue is equivalent to setting the corresponding variable to \texttt{false}.

      Another gadget that we use is the inverter gadget, shown in Figure~\ref {fig:pdocs-var}. This gadget consists of two white vertices, $k$ red vertices, and $k$ blue vertices.  Each colored vertex is connected to both white vertices.  This gadget ensures that one of the white vertices must be colored red and the other one blue, because otherwise there will be $k$ components in the output.  To ensure that this is unacceptable for any optimal solution, we make \(k\) at least as large as the number of gadgets that we use in our construction.

      \frank{We use now the name "inverter gadget" for slightly different gadgets.
              Should we explain the differences? }
\rodrigo{Well, I guess that would be nice}

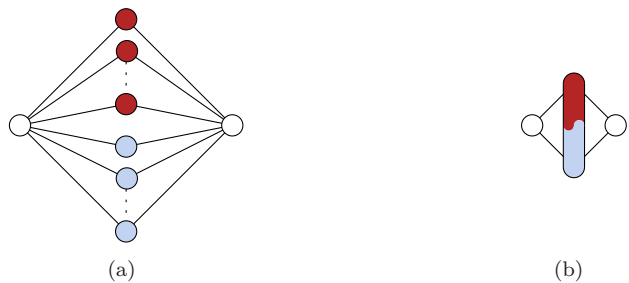
      \tweeplaatjes {pdocs-var} {pdocs-var-symbol} {(a) An inverter, consisting of $k$ red and blue vertices. (b) Symbolic representation of an inverter.}

A clause gadget is a collection of three inverter gadgets,
as well as four extra red vertices.  These are all connected as shown in Figure~\ref{fig:pdocs-clause}.  The red vertices form one large component as long as at least one of the white vertices adjacent to the central red vertex is colored red.

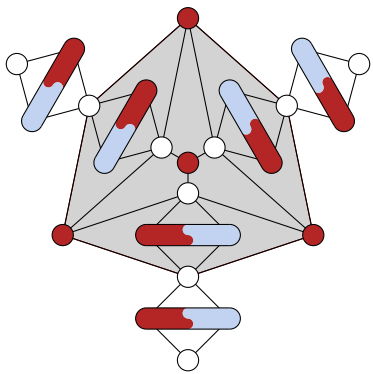
      \eenplaatje {pdocs-clause} {A clause gadget (shown with gray background) connects three inverter gadgets using four extra red vertices.}

Finally, we create edge gadgets to connect variable gadgets to clause gadgets.  An edge gadget is simply a chain of inverter gadgets.  See Figure~\ref{fig:pdocs-chain}.  If a variable \(v\) is negated in clause \(c\), then we replace the edge \((v, c)\) with a chain of an odd number of inverter gadgets, otherwise, we use an even-length chain.  Since the number of inverter gadgets between a variable gadget and one of the clause gadgets that it is connected to determines the color of the final white vertex in the chain, we can see that coloring a variable gadget red corresponds to the final white vertex in a chain to a clause in which that vertex is not negated being colored red.  This implies that coloring a variable gadget red is equivalent to setting its value to \texttt{true}, and that coloring a variable gadget blue is equivalent to setting its value to \texttt{false}.

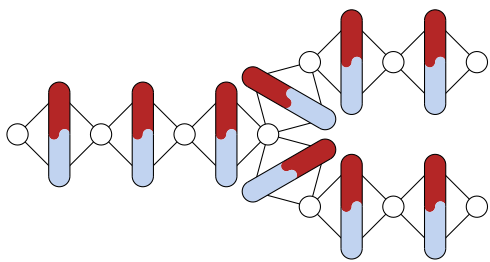
      \eenplaatje {pdocs-chain} {We can chain inverter gadgets. The white vertices must always be colored alternately red and blue.}

      The total number of connected components is equal to the number of white vertices in the construction, minus 2 per clause since the red components are connected, plus the number of unsatisfied clauses. Hence, minimizing the number of connected components involves determining whether the 3-SAT clause can be satisfied, which proves the following.

      \begin {theorem} \label {thm:gphard}
        \gp\ is NP-hard.
      \end {theorem}

\subsection {\pdcs is NP-hard}

  We now show how the construction above can be extended to show that also the more specialized problem \pdcs is NP-hard. The main difference is that we must ensure that at the end of the construction, all blue components become connected into one large blue component, and all red components become connected into one large red component.

  First of all, we need some property of triangulated graphs.
  Let $G = (V, E)$ be a {\em planar triangulated graph}, i.e., a graph embedded in the plane such that all faces of $G$, except the outer face, are triangles.
  Moreover, let us color the vertices of $G$ red and blue such that
  the red vertices $R$ on the outer face of $G$ form more than one single
  connected component and the same is true for the blue vertices $B$.
  %
  We say that the subgraph of $G$ induced by $R$ contains a \emph {proper loop} if it contains a cycle that has at least one vertex of $V \setminus R$ inside.  This leads to the following well-known observation.  See, for example, the book by West~\cite {graphbook} for a proof.

  { \observation \label {obs:tricomp}
    If the subgraph of $G$ induced by $R$ has no proper loops, then the subgraph of $G$ induced by $V \setminus R$ is connected.
  }

  Now, let $G$ again be a planar triangulated graph, and suppose that all vertices of $G$ are colored either red or blue, so $V = R \cup B$, and suppose further that every red or blue component has at least one vertex on the outer face.
  Figure~\ref {fig:trigraph-inner} shows such a graph.
  Then obviously neither $R$ nor $B$ has a proper loop.

  This means that whenever we have such a graph with a sufficient
  number of layers of white vertices around it,
  then we can color it such that we get only one large red component and only one large blue component.

  { \lemma \label {lem:completecoloring}
    Let $G$ be as before, and let $G'$ be a larger graph that contains $G$ and has 2 extra layers of white vertices around $G$, each at least as large as the outer face of $G$.
    Then we can color the white vertices of $G'$ red or blue such that $G'$ has only one red and one blue component.
  }

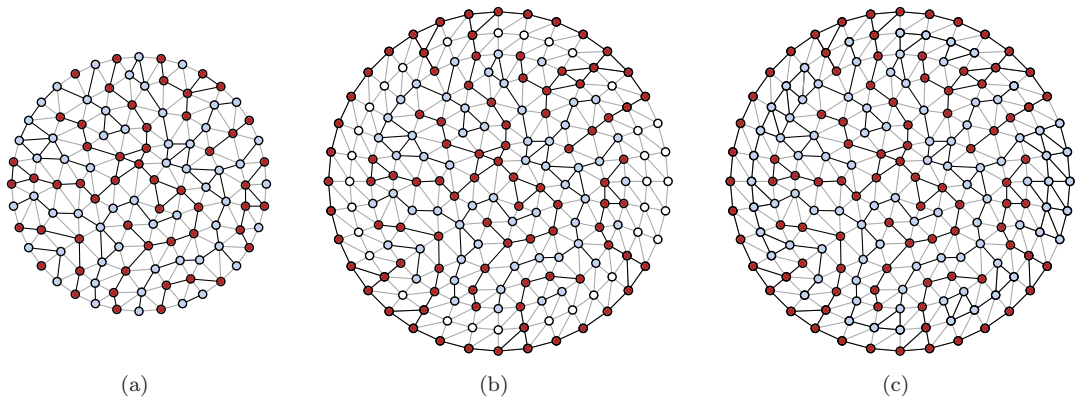
  \drieplaatjes {trigraph-inner} {trigraph-red} {trigraph-final} {(a) A triangulated graph, colored such that every component has a vertex on the outer face. (b) The same graph, augmented with two extra layers of white vertices. The red components are connected into a large components without proper loops. (d)
  All the remaining white vertices can be colored blue, making the blue component also connected.}

\begin{proof}
    We know that all components have at least one vertex on the outer face. For each red component, we select exactly one such vertex and color its counterpieces on the two extra layers also red. Then we color the vertices of the outer layer red to connect all the red components into one large component, as shown in Figure~\ref {fig:trigraph-red}.

    By doing this, we cannot create any proper loops because we only took one vertex from each red component, and because of the regular structure of the two outer layers. Therefore, the red component does not have proper loops, so by Observation~\ref {obs:tricomp} the complement is connected. Hence, we can color the complement blue to obtain a valid coloring.
\end{proof}

  We now show how the construction in the previous section can be extended to show that not only \pdmcs, but also \pdcs is NP-hard. To do this, we must make a construction such that, when the SAT formula is satisfiable, all red components can be connected into one large red component, and all the blue components can be connected into one large blue component. If the formula is not satisfiable, this should not be possible.

  First, it is well known that a planar 3-SAT graph can be embedded in the plane having all variable nodes on one vertical line, and the clause nodes on both sides, as in Figure~\ref {fig:pdocs-layout}. From this, we observe that, if we replace the variable nodes and incident edges by polygonal trees, we can make the embedding such that all faces of the graph are $x$-monotone polygons (except for the outer face), as shown in Figure~\ref {fig:pdocs-monotone} for the example graph.
  This implies a partial ordering on the faces of the graph, based on their above-below relation.\footnote {In fact we don't really need them to be $x$-monotone polygons, any embedding with a directed dual graph that induces a partial ordering such that all faces at the north of the ordering are on the outside of the construction would be good enough for the argument.}

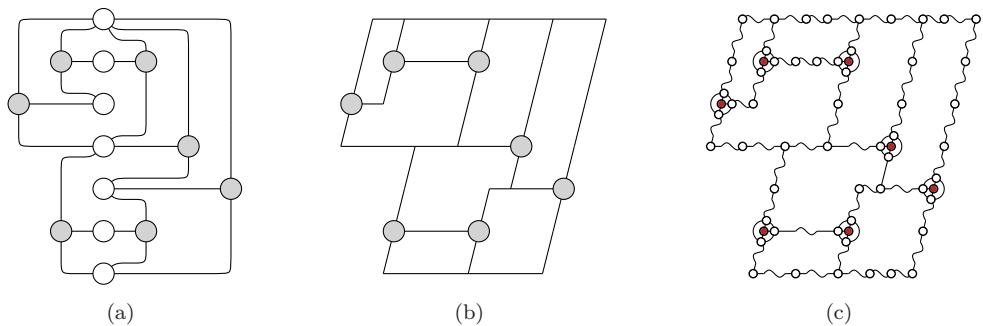
  \drieplaatjes {pdocs-layout} {pdocs-monotone} {pdocs-embedded} {(a) A layout of a planar 3-SAT instance, where the variable nodes (white) are aligned on a single vertical line, and the clause nodes (gray) are on both sides of the line. (b) Polygonal subdivision into $x$-monotone polygons as a result of replacing the variable nodes and adjacent edges by polygonal trees. (c) The construction embedded onto the subdivision. The inverter gadgets are indicated by wiggling edges. Note that there is some freedom in the construction as to how many vertices and inverter gadgets are placed on the edges; only the parity matters.}

  We now replace the polygonal trees by chains of inverter gadgets, as in the previous section. We also replace the clause nodes by clause gadgets as before, except that we do not include the three extra red vertices.  Instead, we connect the neighboring vertices
  of each clause. Figure~\ref {fig:pdocs-bareclause} shows the simplified clause gadget. Figure~\ref {fig:pdocs-embedded} shows an example of the resulting embedding.

  \eenplaatje {pdocs-bareclause} {A clause gadget, simplified.}

  The resulting construction has many white vertices on the boundaries of the $x$-monotone polygons, which end up being colored either red or blue. Furthermore, all components in a final coloring contain at least one of these white vertices, except for those in the unsatisfied clauses. So, what remains to be done is make sure that these white vertices can be connected into one large red and one large blue component, no matter how they are colored.

  To do this, consider one $x$-monotone polygon and its white vertices, as in Figure~\ref {fig:pdocs-face}. An $x$-monotone polygon has two $x$-monotone polygonal chains that both connect the westernmost point to the easternmost point. We assume that on both of these chains there are at least 2 white vertices. Furthermore, we assume that the northern chain of any polygon has at least as many white vertices as the southern chain. If any of these assumptions is not satisfied, we simply add more gadgets to the chains: adding two inverter gadgets into a chain does not change the properties of the construction, and the above/below relations of the polygons define a partial order on them so this process will end.
  We then add edges to the interior of the polygon, connecting every white vertex on the southern chain to two adjacent white vertices on the northern chain that are on both sides of an inverter gadget, as in Figure~\ref {fig:pdocs-face-connections}. This means that whatever the color of a vertex of the southern chain is, it is always connected to at least one vertex of the same color on the northern chain. By induction, this means that every white vertex in the whole construction will be connected to some vertex on the north of the construction that has the same color.
  Finally, we triangulate the polygon (or rather, the graph of the white vertices involved in the polygon) by adding arbitrary edges if necessary, see Figure~\ref {fig:pdocs-face-triang}.
  We call the resulting graph $G$.

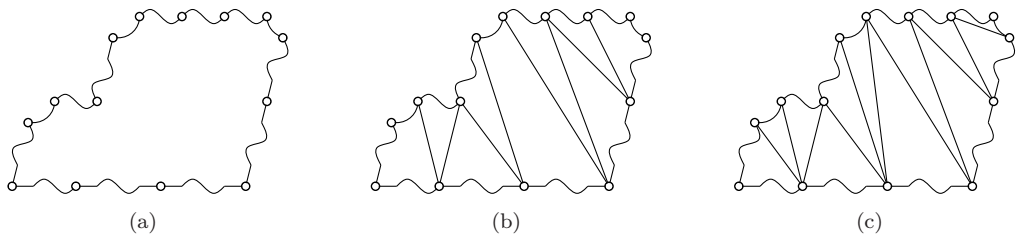
  \drieplaatjes {pdocs-face} {pdocs-face-connections} {pdocs-face-triang} {(a) An $x$-monotone polygon with some white vertices and inverter gadgets on its boundary. (b) Each white vertex on the south has been connected to two white vertices at the north that are separated by an inverter gadget. (c) Some additional edges are inserted to make the graph formed by the white vertices triangulated.}

  To prove that the construction is indeed colorable with two components if the 3-SAT formula is satisfiable, first consider the coloring that makes all clause gadgets satisfied. Then remove the clause vertices (leaving only the three white vertices and the triangle of edges connecting them), and replace all inverter gadgets by edges. The resulting graph is triangulated, and by the above argument, each component is connected to at least one vertex on the outside.
  Conversely, if the 3-SAT formula is not satisfiable, it is not possible to color the construction with two colors such that the vertices with equal colors form two connected components because one of the clauses cannot be satisfied.  This implies that at least one of the clause gadgets cannot be properly colored.
  Finally, by Lemma~\ref {lem:completecoloring}, there also exists a corresponding graph $G'$ that can be colored in two \emph {connected} colors if and only if the 3-SAT formula is satisfyable.
  We conclude:

  { \theorem \label {thm:pdcshard}
    \pdcs is NP-hard.
  }

\section{Discussion}
  We have studied the complexity of removing local extrema from imprecise terrains. We conclude that this complexity changes dramatically between the problem of removing only one kind of extrema (either local minima or local maxima) and the problem of removing both at the same time.
  When one is interested in removing only either local minima or local maxima, this problem can be solved efficiently in $O (n \log n)$ time.
  This problem has real applications in, for example, hydrology, and
  we believe our solution is both simple and practical.
  On the other hand, removing local extrema is hard to approximate, even within a factor of $O(\log \log n)$.

  In addition, we show that even a simplified version of the problem, where vertices have only two possible heights, is already NP-hard.
  This hardness proof exploits a relation to \pdcs, which we also prove is NP-hard.
  To the best of our knowledge, this constitutes the first result for planar graphs for the popular \disconsub\ problem.

  We believe the main remaining open question is whether any constructive results for the problem of minimizing local extrema is possible.
  Though our hardnes result shows there is no hope for even an approximation algorithm with a practical approximation factor in the general case,
  it is still possible that something better can be done for special classes of imprecise terrains.
  Moreover, Equation~\ref {eq:gap} suggests that in real terrains, it may be much easier to remove local extrema than what our theoretical results claim.
  It would be interesting to run experiments on real terrains in order to analyze the size of this gap in practice, and to investigate if there is some set of  \emph{realistic} conditions on the input terrain for which the gap is small.


\section*{Acknowledgments}

{\small
We thank Jeff Phillips for proposing the problem in Section~\ref{sec:extrema}.
We thank an anonymous reviewer for useful suggestions that improved the paper.
C.G. is funded by the German Ministry for Education and Research
(BMBF) under grant number 03NAPI4 ``ADVEST''.
M.L. is funded by the U.S. Office of Naval Research under grant N00014-08-1-1015.
R.I.S. is supported by the Netherlands
Organisation for
Scientific Research (NWO).
}

\small

\bibliographystyle {abbrv}
\bibliography {refs}

\end{document}